\documentclass[a4paper,english,cleveref, autoref]{lipics-v2019}
\newdimen\proofrulebreadth \proofrulebreadth=.05em
\newdimen\proofdotseparation \proofdotseparation=1.25ex
\newdimen\proofrulebaseline \proofrulebaseline=2ex
\newcount\proofdotnumber \proofdotnumber=3
\let\then\relax
\def\hfi{\hskip0pt plus.0001fil}
\mathchardef\squigto="3A3B
\newif\ifinsideprooftree\insideprooftreefalse
\newif\ifonleftofproofrule\onleftofproofrulefalse
\newif\ifproofdots\proofdotsfalse
\newif\ifdoubleproof\doubleprooffalse
\let\wereinproofbit\relax
\newdimen\shortenproofleft
\newdimen\shortenproofright
\newdimen\proofbelowshift
\newbox\proofabove
\newbox\proofbelow
\newbox\proofrulename
\def\shiftproofbelow{\let\next\relax\afterassignment\setshiftproofbelow\dimen0 }
\def\shiftproofbelowneg{\def\next{\multiply\dimen0 by-1 }%
\afterassignment\setshiftproofbelow\dimen0 }
\def\setshiftproofbelow{\next\proofbelowshift=\dimen0 }
\def\setproofrulebreadth{\proofrulebreadth}

\def\prooftree{% NESTED ZERO (\ifonleftofproofrule)
\ifnum  \lastpenalty=1
\then   \unpenalty
\else   \onleftofproofrulefalse
\fi
\ifonleftofproofrule
\else   \ifinsideprooftree
        \then   \hskip.5em plus1fil
        \fi
\fi
\bgroup% NESTED ONE (\proofbelow, \proofrulename, \proofabove,
\setbox\proofbelow=\hbox{}\setbox\proofrulename=\hbox{}%
\let\justifies\proofover\let\leadsto\proofoverdots\let\Justifies\proofoverdbl
\let\using\proofusing\let\[\prooftree
\ifinsideprooftree\let\]\endprooftree\fi
\proofdotsfalse\doubleprooffalse
\let\thickness\setproofrulebreadth
\let\shiftright\shiftproofbelow \let\shift\shiftproofbelow
\let\shiftleft\shiftproofbelowneg
\let\ifwasinsideprooftree\ifinsideprooftree
\insideprooftreetrue
\setbox\proofabove=\hbox\bgroup$\displaystyle % NESTED TWO
\let\wereinproofbit\prooftree
\shortenproofleft=0pt \shortenproofright=0pt \proofbelowshift=0pt
\onleftofproofruletrue\penalty1
}

\def\eproofbit{% NESTED TWO
\ifx    \wereinproofbit\prooftree
\then   \ifcase \lastpenalty
        \then   \shortenproofright=0pt  % 0: some other object, no indentation
        \or     \unpenalty\hfil         % 1: empty hypotheses, just glue
        \or     \unpenalty\unskip       % 2: just had a tree, remove glue
        \else   \shortenproofright=0pt  % eh?
        \fi
\fi
\global\dimen0=\shortenproofleft
\global\dimen1=\shortenproofright
\global\dimen2=\proofrulebreadth
\global\dimen3=\proofbelowshift
\global\dimen4=\proofdotseparation
\global\count255=\proofdotnumber
$\egroup  % NESTED ONE
\shortenproofleft=\dimen0
\shortenproofright=\dimen1
\proofrulebreadth=\dimen2
\proofbelowshift=\dimen3
\proofdotseparation=\dimen4
\proofdotnumber=\count255
}

\def\proofover{% NESTED TWO
\eproofbit % NESTED ONE
\setbox\proofbelow=\hbox\bgroup % NESTED TWO
\let\wereinproofbit\proofover
$\displaystyle
}%
\def\proofoverdbl{% NESTED TWO
\eproofbit % NESTED ONE
\doubleprooftrue
\setbox\proofbelow=\hbox\bgroup % NESTED TWO
\let\wereinproofbit\proofoverdbl
$\displaystyle
}%
\def\proofoverdots{% NESTED TWO
\eproofbit % NESTED ONE
\proofdotstrue
\setbox\proofbelow=\hbox\bgroup % NESTED TWO
\let\wereinproofbit\proofoverdots
$\displaystyle
}%
\def\proofusing{% NESTED TWO
\eproofbit % NESTED ONE
\setbox\proofrulename=\hbox\bgroup % NESTED TWO
\let\wereinproofbit\proofusing
\kern0.3em$
}

\def\endprooftree{% NESTED TWO
\eproofbit % NESTED ONE
  \dimen5 =0pt% spread of hypotheses
\dimen0=\wd\proofabove \advance\dimen0-\shortenproofleft
\advance\dimen0-\shortenproofright
\dimen1=.5\dimen0 \advance\dimen1-.5\wd\proofbelow
\dimen4=\dimen1
\advance\dimen1\proofbelowshift \advance\dimen4-\proofbelowshift
\ifdim  \dimen1<0pt
\then   \advance\shortenproofleft\dimen1
        \advance\dimen0-\dimen1
        \dimen1=0pt
        \ifdim  \shortenproofleft<0pt
        \then   \setbox\proofabove=\hbox{%
                        \kern-\shortenproofleft\unhbox\proofabove}%
                \shortenproofleft=0pt
        \fi
\fi
\ifdim  \dimen4<0pt
\then   \advance\shortenproofright\dimen4
        \advance\dimen0-\dimen4
        \dimen4=0pt
\fi
\ifdim  \shortenproofright<\wd\proofrulename
\then   \shortenproofright=\wd\proofrulename
\fi
\dimen2=\shortenproofleft \advance\dimen2 by\dimen1
\dimen3=\shortenproofright\advance\dimen3 by\dimen4
\ifproofdots
\then
        \dimen6=\shortenproofleft \advance\dimen6 .5\dimen0
        \setbox1=\vbox to\proofdotseparation{\vss\hbox{$\cdot$}\vss}%
        \setbox0=\hbox{%
                \advance\dimen6-.5\wd1
                \kern\dimen6
                $\vcenter to\proofdotnumber\proofdotseparation
                        {\leaders\box1\vfill}$%
                \unhbox\proofrulename}%
\else   \dimen6=\fontdimen22\the\textfont2 % height of maths axis
        \dimen7=\dimen6
        \advance\dimen6by.5\proofrulebreadth
        \advance\dimen7by-.5\proofrulebreadth
        \setbox0=\hbox{%
                \kern\shortenproofleft
                \ifdoubleproof
                \then   \hbox to\dimen0{%
                        $\mathsurround0pt\mathord=\mkern-6mu%
                        \cleaders\hbox{$\mkern-2mu=\mkern-2mu$}\hfill
                        \mkern-6mu\mathord=$}%
                \else   \vrule height\dimen6 depth-\dimen7 width\dimen0
                \fi
                \unhbox\proofrulename}%
        \ht0=\dimen6 \dp0=-\dimen7
\fi
\let\doll\relax
\ifwasinsideprooftree
\then   \let\VBOX\vbox
\else   \ifmmode\else$\let\doll=$\fi
        \let\VBOX\vcenter
\fi
\VBOX   {\baselineskip\proofrulebaseline \lineskip.2ex
        \expandafter\lineskiplimit\ifproofdots0ex\else-0.6ex\fi
        \hbox   spread\dimen5   {\hfi\unhbox\proofabove\hfi}%
        \hbox{\box0}%
        \hbox   {\kern\dimen2 \box\proofbelow}}\doll%
\global\dimen2=\dimen2
\global\dimen3=\dimen3
\egroup % NESTED ZERO
\ifonleftofproofrule
\then   \shortenproofleft=\dimen2
\fi
\shortenproofright=\dimen3
\onleftofproofrulefalse
\ifinsideprooftree
\then   \hskip.5em plus 1fil \penalty2
\fi
}

\title{The differential calculus of causal functions}

\titlerunning{The differential calculus of causal functions}

\author{David Sprunger}{National Institute of Informatics, Tokyo JP}{sprunger@nii.ac.jp}{}{This author is supported by ERATO HASUO Metamathematics for Systems Design Project (No.\ JPMJER1603), JST.}
\author{Bart Jacobs}{Radboud University, Nijmegen NL}{bart@cs.ru.nl}{}{}
\authorrunning{D. Sprunger and B. Jacobs}%TODO mandatory. First: Use abbreviated first/middle names. Second (only in severe cases): Use first author plus 'et al.'

\Copyright{David Sprunger and Bart Jacobs}%TODO mandatory, please use full first names. LIPIcs license is "CC-BY";  http://creativecommons.org/licenses/by/3.0/

\ccsdesc[100]{Mathematics of computing~Differential calculus}
\ccsdesc[100]{Computing methodologies~Neural networks}%TODO mandatory: Please choose ACM 2012 classifications from https://dl.acm.org/ccs/ccs_flat.cfm 

\keywords{sequences, causal functions, derivatives, recurrent neural networks, Elman networks}%TODO mandatory; please add comma-separated list of keywords

\category{}%optional, e.g. invited paper

\relatedversion{}%optional, e.g. full version hosted on arXiv, HAL, or other respository/website
\supplement{}%optional, e.g. related research data, source code, ... hosted on a repository like zenodo, figshare, GitHub, ...

\nolinenumbers %uncomment to disable line numbering

\EventEditors{John Q. Open and Joan R. Access}
\EventNoEds{2}
\EventLongTitle{8th Conference on Very Important Topics}
\EventShortTitle{CVIT 2019}
\EventAcronym{CVIT}
\EventYear{2019}
\EventDate{June 4--7, 2019}
\EventLocation{London, United Kingdom}
\EventLogo{}
\SeriesVolume{42}
\ArticleNo{23}
\newcommand{\R}{\mathbb{R}}

\newcommand{\w}{\omega}

\newcommand{\x}{\times}

\newcommand{\pdt}[2]{\frac{\partial #1}{\partial #2}}
\newcommand{\tuple}[1]{\langle#1\rangle}
\newcommand{\lsem}{[\![}
\newcommand{\rsem}{]\!]}
\newcommand{\sem}[1]{\lsem #1\rsem}

\newcommand{\Bf}[1]{{\bf #1}}

\newcommand{\mc}{\mathcal}

\newcommand{\Set}{\Bf{Set}}

\newcommand{\final}[1]{{!_{#1}}}

\newcommand{\teq}{\mathbin{\triangleq}}

\def\squD{\mc D}    % read: 'square differential (operator)'
\def\seqD{\squD^*}  % read: 'sequence differential (operator'
\def\hd{{\tt hd}}
\def\tl{{\tt tl}}
\def\id{{\rm id}}

\def\rec{{\tt rec}}
\def\map{{\tt map}}

\def\inv{^{-1}}

\begin{document}

\maketitle

\begin{abstract}

Causal functions of sequences occur throughout computer science, from theory
to hardware to machine learning. Mealy machines, synchronous digital circuits,
signal flow graphs, and recurrent neural networks all have behaviour that can
be described by causal functions. In this work, we examine a differential
calculus of causal functions which includes many of the familiar properties of
standard multivariable differential calculus. These causal functions operate
on infinite sequences, but this work gives a different notion of an
infinite-dimensional derivative than either the Fréchet or Gateaux derivative
used in functional analysis. In addition to showing many standard properties
of differentiation, we show causal differentiation obeys a unique recurrence
rule. We use this recurrence rule to compute the derivative of a simple
recurrent neural network called an Elman network by hand and describe how the
computed derivative can be used to train the network.

\end{abstract}

\section{Introduction}

Many computations on infinite data streams operate in a \emph{causal} manner,
meaning their $k$\textsuperscript{th} output depends only on the first $k$
inputs. Mealy machines, clocked digital circuits, signal flow graphs,
recurrent neural networks, and discrete time feedback loops in control theory
are a few examples of systems performing such computations. When designing
these kinds of systems to fit some specification, a common issue is figuring
out how adjusting one part of the system will affect the behaviour of the
whole. If the system has some real-valued semantics, as is especially common
in machine learning or control theory, the derivative of these semantics with
respect to a quantity of interest, say an internal parameter, gives a
locally-valid first-order estimate of the system-wide effect of a small change
to that quantity. Unfortunately, since the most natural semantics for infinite
data streams is in an infinite-dimensional vector space, it is not practical
to use the resulting infinite-dimensional derivative.

To get around this, one tactic is to replace the infinite system by a finite
system obtained by an approximation or heuristic and take derivatives of the
replacement system. This can be seen, for example, in \emph{backpropagation
through time}~\cite{bptt}, which trains a recurrent neural network by first
unrolling the feedback loop the appropriate number of times and then applying
traditional backpropagation to the unrolled network.

This tactic has the advantage that we can take derivatives in a familiar
(finite-dimensional) setting, but the disadvantage that it is not clear what
properties survive the approximation process from the unfamiliar
(infinite-dimensional) setting. For example, it is not immediately clear
whether backpropagation through time obeys the usual rules of differential
calculus, like a sum or chain rule, nor is this issue confronted in the
literature, to the best of our knowledge. Thus, useful compositional
properties of differentiation are ignored in exchange for a comfortable
setting in which to do calculus.

In this work, we take advantage of the fact that causal functions between
sequences are already essentially limits of finite-dimensional functions and
therefore have derivatives which can also be expressed as essentially limits
of the derivatives of these finite-dimensional functions. This leads us to the
basics of a differential calculus of causal functions. Unlike \emph{arbitrary}
functions between sequences, this limiting process allows us to avoid the use
of normed vector spaces, and so we believe our notion of derivative is
distinct from Fréchet derivatives.

\textbf{Outline.} In section~\ref{sec:causalDef}, we define causal functions
and recall several mechanisms by which these functions on infinite data can be
defined. In particular, we recall a coalgebraic scheme finding causal
functions as the behaviour of Mealy machines
(proposition~\ref{prop:coalgebraicDef}), and give a definitional scheme in
terms of so-called \emph{finite approximants}
(definition~\ref{def:finiteApproxs}). In section~\ref{sec:causalDiff}, we
define differentiability and derivatives of causal functions on real-vector
sequences (definition~\ref{def:causalDiff}) and compute several examples. In
section~\ref{sec:rules}, we obtain several rules for our differential causal
calculus analogous to those of multivariable calculus, including a chain rule,
parallel rule, sum rule, product rule, reciprocal rule, and quotient rule
(propositions \ref{prop:causalChainRule}, \ref{prop:causalParallelRule},
\ref{prop:causalSumRule}, \ref{prop:causalProductRule},
\ref{prop:causalReciprocalRule}, and \ref{prop:causalQuotientRule},
respectively). We additionally find a new rule without a traditional analogue
we call the recurrence rule (theorem~\ref{thm:causalRecurrenceRule}). Finally,
in section~\ref{sec:elman}, we apply this calculus to find derivatives of a
simple kind of recurrent neural network called an Elman
network~\cite{Elman_1990} by hand. We also demonstrate how to use the
derivative of the network with respect to a parameter to guide updates of that
parameter to drive the network towards a desired behaviour.

%%%%%%%%%%%%%%%%%%%%%%%%%%%%%%%%%%%%%%%%%%%%%%%%%%%%%%%%%%%%%%%%%%%%%%%
%% SECTION: Causal functions                                         %%
%%%%%%%%%%%%%%%%%%%%%%%%%%%%%%%%%%%%%%%%%%%%%%%%%%%%%%%%%%%%%%%%%%%%%%%
%%%%%%%%%%%%%%%%%%%%%%%%%%%%%%%%%%%%%%%%%%%%%%%%%%%%%%%%%%%%%%%%%%%%%%%
\section{Causal functions of sequences}\label{sec:causalDef}

A \textit{sequence} or \textit{stream} in a set $A$ is a countably infinite
list of values from $A$, which we also think of as a function from the natural
numbers $\w$ to $A$. If $\sigma$ is a stream in $A$, we denote its value at $k
\in \w$ by $\sigma_k$. We may also think of a stream as a listing of its
image, like $\sigma = (\sigma_0, \sigma_1, \ldots)$. The set of all sequences
in $A$ is denoted $A^\w$.

Given $a \in A$ and $\sigma \in A^\w$, we can form a new sequence by
prepending $a$ to $\sigma$. The sequence $a:\sigma$ is defined by
$(a:\sigma)_0 = a$ and $(a:\sigma)_{k+1} = \sigma_k$. This operation can be
extended to prepend arbitrary finite-length words $w \in A^*$ by the obvious
recursion. Conversely, we can destruct a given sequence into an element and a
second sequence with functions $\hd: A^\w \to A$ and $\tl: A^\w \to A^\w$
defined by $\hd(\sigma) = \sigma_0$ and $\tl(\sigma)_k = \sigma_{k+1}$.

\begin{definition}[slicing]
  If $\sigma \in A^\w$ is a stream and $j \leq k$ are natural numbers, the
  \emph{slicing} $\sigma_{j:k}$ is the list $(\sigma_j, \sigma_{j+1}, \ldots,
  \sigma_k) \in A^{k-j+1}$.
\end{definition}

\begin{definition}[causal function]
  A function $f: A^\w\to B^\w$ is \emph{causal} means $\sigma_{0:k} =
  \tau_{0:k}$ implies $f(\sigma)_{0:k} = f(\tau)_{0:k}$ for all $\sigma, \tau
  \in A^\w$ and $k \in \w$.
\end{definition}

\subsection{Causal functions via coalgebraic finality}

A standard coalgebraic approach to causal functions is to view them as the
behaviour of Mealy machines.

\begin{definition}[Mealy functor]
  Given two sets $A, B$, the functor $M_{A,B}: \Set\to\Set$ is defined by
  $M_{A,B}(X) = (B\x X)^A$ on objects and $M_{A,B}(f): \phi\mapsto (\id_B\x f)
  \circ \phi$ on morphisms.
\end{definition}

$M_{A, B}$-coalgebras are Mealy machines with input alphabet $A$ and output
alphabet $B$, and possibly an infinite state space. The set of causal
functions $A^\w\to B^\w$ carries a final $M_{A, B}$-coalgebra using the
following operations, originally observed by Rutten in \cite{ruttenMealy}.

\begin{definition}
  The \emph{Mealy output} of a causal function $f: A^\w\to B^\w$ is the function
  $\hd f: A \to B$ defined by $(\hd f)(a) = f(a:\sigma)_0$ for any $\sigma \in
  A^\w$.
\end{definition}

\begin{definition}
  Given $a \in A$ and a causal function $f: A^\w\to B^\w$, the \emph{Mealy
  ($a$-)derivative} of $f$ is the causal function $\partial_a f: A^\w\to B^\w$
  defined by $(\partial_a f)(\sigma) = \tl(f(a:\sigma))$.
\end{definition}

Note $\hd(f)$ is well-defined even though $\sigma$ may be freely chosen due to
the causality of $f$.

\begin{proposition}[Proposition 2.2, \cite{ruttenMealy}]\label{prop:coalgebraicDef}
  The set of causal functions $A^\w\to B^\w$ carries an $M_{A,B}$-coalgebra
  via $f \mapsto \lambda a.((\hd f)(a), \partial_a f)$, which is a final
  $M_{A, B}$-coalgebra.
\end{proposition}

Hence, a coalgebraic methodology for defining causal functions is to define a
Mealy machine and take the image of a particular state in the final coalgebra.
By constructing the Mealy machine cleverly, one can ensure the resulting
causal function has some desired properties. This is the core idea behind the
``syntactic method'' using GSOS definitions in \cite{sdes}. In that work, a
Mealy machine of terms is built in such a way that all causal functions
$(A^k)^\w \to A^\w$ can be recovered.

\begin{example}\label{ex:vectorSpace}
  Suppose $(A, +_A, \cdot_A, 0_A)$ is a vector space over $\R$. This vector
  space structure can be extended to $A^\w$ componentwise in the obvious way.
  To illustrate the coalgebraic method, we characterise this structure with
  coalgebraic definitions.

  To define sequence vector sum coalgebraically, we define a Mealy machine $1
  \to (A\x 1)^{A\x A}$ with one state, satisfying $\hd(s)(a, a') = a +_A a'$
  and $\partial_{(a, a')}(s) = s$. Then ${+_{A^\w}}: (A\x A)^\w \to A^\w$ is
  defined to be the image of $s$ in the final $M_{A^2, A}$-coalgebra.

  Note that technically the vector sum in $A^\w$ should be a function of type
  $A^\w\x A^\w \to A^\w$, so we are tacitly using the isomorphism between
  $(A\x A)^\w$ and $A^\w\x A^\w$. We will be using similar recastings of
  sequences in the sequel without bringing up this point again.

  The zero vector can similarly be defined by a single state Mealy machine $1
  \to (A\x 1)^1$ with input alphabet 1 and output alphabet $A$, satisfying
  $\hd(s')(*) = 0_A$ and $\partial_*(s') = s'$. The zero vector of $A^\w$ is
  the global element picked out by the image of $s'$.

  Finally, scalar multiplication can be defined with a Mealy machine
  $\R\to(A\x\R)^A$ with states $r \in \R$, such that $\hd(r)(a) =
  r\cdot_A a$ and $\partial_a r = r$. Then $r\cdot_{A^\w} \sigma \teq
  \sem{r}(\sigma)$, where $\sem{r}$ is the image of $r$ in the final
  $M_{A, A}$-coalgebra.

  We immediately begin dropping the subscripts from $+_{A^\w}$ and
  $\cdot_{A^\w}$ and when the relevant vector space can be inferred from
  context.
\end{example}

\subsection{Causal functions via finite approximation}

Another approach to causal functions is consider them as a limit of finite
approximations, replacing the single function on infinite data with infinitely
many functions on finite data. There are (at least) two approaches with this
general style, which we briefly describe next.

\begin{definition}\label{def:finiteApproxs}
  Let $f: A^\w\to B^\w$ be a causal function and $\sigma \in A^\w$.

  The \emph{pointwise approximation} of $f$ is the sequence of
  functions $U_k(f): A^{k+1} \to B$ defined by $U_k(f)(w) \teq f(w:\sigma)_k$.

  The \emph{stringwise approximation} of $f$ is the sequence of
  functions $T_k(f): A^{k+1} \to B^{k+1}$ defined by $T_k(f)(w) \teq
  f(w:\sigma)_{0:k}$.
\end{definition}

Again, these are well-defined despite $\sigma$ being arbitrary due to $f$'s
causality. We chose the letters $U$ and $T$ deliberately---sometimes the
pointwise approximants of a causal function are called its \emph{U}nrollings,
and the stringwise approximants are called its \emph{T}runcations.

Conversely, given an arbitrary collection of functions $u_k: A^{k+1} \to B$
for $k \in \w$, there is a unique causal function whose pointwise
approximation is the sequence $u_k$. Thus we have the following
bijective correspondence:
\begin{equation}
\label{eqn:caualapprox}
\begin{prooftree}
A^\w\longrightarrow B^\w \mbox{ causal}
\Justifies
A^{k+1} \longrightarrow B \mbox{ for each }k\in\w
\end{prooftree}
\end{equation}

We can nearly do the same for stringwise approximations, but the sequence
$t_k: A^{k+1} \to B^{k+1}$ must satisfy $t_k(w) = t_{k+1}(wa)_{0:k}$ for all
$w \in A^{k+1}$ and $a \in A$.

The interchangeability between a causal function and its approximants
is a crucial theme in this work. Since a function's pointwise and
stringwise approximants are inter-obtainable, we will sometimes
refer to a causal function's ``finite approximants'' by which we mean
either family of approximants.

\subsection{Causal functions via recurrence}

Finite approximants are a very flexible way of defining causal functions, but
causal functions may have a more compact representation when they conform to a
regular pattern. Recurrence is one such pattern where a causal function is
defined by repeatedly using an ordinary function $g: A \x B \to B$ and an
initial value $i \in B$ to obtain $\rec_i(g): A^\w \to B^\w$ via:%
$$[\rec_i(g)(\sigma)]_k = \begin{cases} g(\sigma_0, i) & \text{ if } k = 0 \\ 
g(\sigma_k, [\rec_i(g)(\sigma)]_{k-1}) & \text{ if } k > 0\end{cases}$$

Recurrent definitions can be converted into finite approximant definitions
using the following: $U_k(\rec_i(g))(\sigma_{0:k}) = g(\sigma_k,
g(\sigma_{k-1}, \ldots g(\sigma_1, g(\sigma_0, i)) \ldots ))$. Note these
pointwise approximants satisfy the recurrence relation
$U_k(\rec_i(g))(\sigma_{0:k}) = g(\sigma_k,
U_{k-1}(\rec_i(g))(\sigma_{0:k-1}))$.

\begin{example}\label{ex:runningProduct}
  The unary running product function $\prod: \R^\w \to \R^\w$ can be defined by
  a recurrence relation:
  $$\prod(\sigma) = \tau \Leftrightarrow \begin{cases}
  \tau_{k+1} = \sigma_{k+1}\cdot\tau_k &\text{ after } \tau_0 = \sigma_0\cdot 1 \end{cases}$$
  Here $g$ is multiplication of reals and $i = 1$. In approximant form,
  $[\prod(\sigma)]_k = \prod_{i = 0}^k \sigma_i$.
\end{example}

A special case of recurrent causal functions occurs when there is an $h: A \to
B$ such that $g(a, b) = h(a)$ for all $(a, b) \in A\x B$. In this case,
$[\rec_i(g)(\sigma)]_k = h(\sigma_k)$ and in particular does not depend on the
initial value $i$ or any entry $\sigma_j$ for $j < k$. We denote $\rec_i(g)$
by $\map(h)$ in this special case since it maps $h$ componentwise across the
input sequence.

\section{Differentiating causal functions}\label{sec:causalDiff}

Our goal in this work is to develop a basic differential calculus for causal
functions. Thus we will focus our attention on causal functions between
real-vector sequences $(\R^n)^\w$ for $n \in \w$, specializing from causal
functions on general sets from the last section. We will draw many of our
illustrating examples for derivatives from Rutten's \emph{stream
calculus}~\cite{Rutten_2005}, which describes many such causal functions
between real-number streams. More importantly, \cite{Rutten_2005} establishes
many useful algebraic properties of these functions rigorously via coalgebraic
methods.

There are many different approaches one might consider to defining
differentiable causal functions. One might be to take the original coalgebraic
definition and replace the underlying category ($\Set$) with a category of
finite-dimensional Cartesian spaces and differentiable (or smooth) maps.
Unfortunately, the space of differentiable functions between
finite-dimensional spaces is not finite-dimensional, so the exponential needed
to define the $M_{A, B}$ functor in this category does not exist.

Another approach is to think of causal functions as functions between infinite
dimensional vector spaces and take standard notions from analysis, like
Fréchet derivatives, and apply them in this context. However, norms on
sequence spaces usually impose a finiteness condition like bounded or
square-summable on the domains and ranges of sequence functions. These
restrictions are compatible with many causal functions like the pointwise sum
function above, but other causal functions like the running product function
become significantly less interesting.

Our approach to differentiating causal functions is to consider a causal
function differentiable when all of its finite approximants are differentiable
via the correspondence~\eqref{eqn:caualapprox}. We will develop this idea
rigorously in section~\ref{ssec:derivativeDefinition}, but first we need to
know a bit about linear causal functions.

\subsection{Linear causal functions}\label{ssec:linearCausal}

Stated abstractly, the derivative of a function at a point is a linear map
which provides an approximate change in the output of a function given an
input representing a small change in the input to that function
\cite{spivak1965calculus}. Since linear functions $\R \to \R$ are in bijective
correspondence with their slopes, typically in single-variable calculus the
derivative of a function at a point is instead given as a single real number.
In multivariable calculus, derivatives are usually represented by (Jacobian)
matrices since matrices represent linear maps between finite dimensional
spaces. Linear functions between infinite dimensional vector spaces do not
have a similarly compact, computationally-useful representation, but we can
still define derivatives of (causal) functions at points to be linear (causal)
maps.

We described the natural vector space structure of $(\R^n)^\w$ in
Example~\ref{ex:vectorSpace}. A linear causal function is a causal function
which is also linear with respect to this vector space structure.

\begin{definition}
  A causal function $f: (\R^n)^\w \to (\R^m)^\w$ is \emph{linear} when
  $f(r\cdot \sigma) = r\cdot f(\sigma)$ and $f(\sigma + \tau) = f(\sigma) +
  f(\tau)$ for all $r \in \R$ and $\sigma, \tau \in (\R^n)^\w$.
\end{definition}

\begin{lemma}\label{lem:linearApproximants}
  Let $f: (\R^n)^\w \to (\R^m)^\w$ be a causal function. The following are equivalent:
  \begin{enumerate}
    \item $f$ is linear,
    \item $U_k(f): (\R^n)^{k+1}\to \R^m$ is linear for all $k \in \w$, and
    \item $T_k(f): (\R^n)^{k+1} \to (\R^m)^{k+1}$ is linear for all $k \in \w$.
  \end{enumerate}
\end{lemma}

This refines the correspondence~\eqref{eqn:caualapprox}, allowing us to define
a linear causal function by naming linear finite approximants.

Since linear functions between finite dimensional vector spaces can
be represented by matrices, we can think of linear causal functions as
limits of the matrices representing its finite approximants. This
view results in row-finite infinite matrices, such as:
$$\begin{bmatrix}
  A_{00} & 0      & 0      & \ldots \\
  A_{10} & A_{11} & 0      & \ldots \\
  A_{20} & A_{21} & A_{22} & \ldots \\
  \vdots & \vdots & \vdots & \ddots
\end{bmatrix}$$
\noindent where the $A_{ij}$ are $m$-row, $n$-column blocks such that for
$j>i$ all entries are 0. These are related to the matrices for the
approximants of the causal function as follows.

\begin{enumerate}
  \item The matrix $\begin{bmatrix}A_{k0} & A_{k1} & \ldots & A_{kk}\end{bmatrix}$
  is the matrix representing $U_k(f)$.
  \item The matrix $\begin{bmatrix}
    A_{00} & 0      & 0      & \ldots & 0      \\
    A_{10} & A_{11} & 0      & \ldots & 0      \\
    \vdots & \vdots & \vdots & \ddots & \vdots \\
    A_{k0} & A_{k1} & A_{k2} & \ldots & A_{kk}
  \end{bmatrix}$ is the matrix representing $T_k(f)$. The compatibility
  conditions on the functions $T_k(f)$ ensure that the matrix for $T_k(f)$ can
  be found in the upper left corner of the matrix for $T_{k+1}(f)$. Note also
  the upper triangular nature of the matrices for $T_k(f)$ are a consequence
  of causality---the first $m$ outputs can depend only on the first $n$ inputs,
  so the last entries in the top row must all be 0 and so on.
\end{enumerate}

Unlike finite-dimensional matrices, we do not think these infinite matrices
are a computationally useful representation, but they are conceptually useful
to get an idea of how causal linear functions can be considered the limit of
their linear truncations.

\subsection{Definition of derivative}\label{ssec:derivativeDefinition}

As we have mentioned, we will use the derivatives of the approximants of a
causal function to define the derivative of the causal function itself. We
denote the $m$-row, $n$-column Jacobian matrix of a differentiable function
$\varphi: \R^n \to \R^m$ at $x \in \R^n$ by $J\varphi(x)$. Recall this matrix
is%
$$\begin{bmatrix}
\pdt{\varphi_1}{x_1}(x) & \pdt{\varphi_1}{x_2}(x) & \ldots & \pdt{\varphi_1}{x_n}(x) \\
\pdt{\varphi_2}{x_1}(x) & \pdt{\varphi_2}{x_2}(x) & \ldots & \pdt{\varphi_2}{x_n}(x) \\
\vdots & \vdots & \ddots & \vdots \\
\pdt{\varphi_m}{x_1}(x) & \pdt{\varphi_m}{x_2}(x) & \ldots & \pdt{\varphi_m}{x_n}(x)
\end{bmatrix}$$
\noindent where $\varphi_i: \R^n \to \R$ and $\varphi = \tuple{\varphi_1, \ldots \varphi_m}$. We will
also be glossing over the distinction between a matrix and the linear function
it represents, using $J\varphi(x)$ to mean either when convenient.

\begin{definition}\label{def:causalDiff}
A causal function $f: (\R^n)^\w\to (\R^m)^\w$ is \emph{differentiable at
$\sigma \in (\R^n)^\w$} if all of its finite approximants $U_k(f):
(\R^n)^{k+1} \to \R^m$ are differentiable at $\sigma_{0:k}$ for all $k \in
\w$. If $f$ is differentiable at $\sigma$, the \emph{derivative of $f$ at
$\sigma$} is the unique linear causal function $\seqD f(\sigma): (\R^n)^\w\to
(\R^m)^\w$ satisfying $U_k(\seqD f(\sigma))= J(U_k(f))(\sigma_{0:k})$.
\end{definition}

In this definition we are using the correspondence~\eqref{eqn:caualapprox},
refined in Lemma~\ref{lem:linearApproximants}, which allows us to define a
causal (linear) function by specifying its (linear) finite approximants. We
could equally well have used stringwise approximants in this definition rather
than pointwise approximants, as the following lemma states.

\begin{lemma}
  The causal function $f$ is differentiable at $\sigma$ if and only if each of
  $T_k(f)$ are differentiable at $\sigma_{0:k}$ for all $k \in \w$. In this case,
  $\seqD f(\sigma)$ satisfies $T_k(\seqD f(\sigma)) = J(T_k(f))(\sigma_{0:k})$.
\end{lemma}

Though we have mentioned this is not particularly useful computationally, the
derivative of a differentiable function at a point has a representation as a
row-finite infinite matrix.

\begin{lemma}
If $f$ is differentiable at $\sigma$, each $U_k(f): (\R^n)^{k+1} \to \R^m$ has
an $m$-row, $n(k+1)$-column Jacobian matrix representing its derivative at
$\sigma_{0:k}$. Let $A_{ki}$ be $m$-row, $n$-column blocks of this Jacobian, so
that $J(U_k(f))(\sigma_{0:k}) = \begin{bmatrix}A_{k0} & A_{k1} & \ldots &
A_{kk}\end{bmatrix}$ The derivative of $f$ at $\sigma$ is the linear
causal function represented by the row-finite infinite matrix
$$\seqD f(\sigma) = \begin{bmatrix}
  A_{00} & 0      & 0      & \ldots \\
  A_{10} & A_{11} & 0      & \ldots \\
  A_{20} & A_{21} & A_{22} & \ldots \\
  \vdots & \vdots & \vdots & \ddots
\end{bmatrix}$$
\end{lemma}

% The   is the linear 
% is the causal function $\seqD f: (\R^n\x\R^n)^\w \to (\R^m)^\w$,
% determined via the correspondence~\eqref{eqn:caualapprox}, by the
% requirement that its $k$th pointwise approximant satisfies:
% \[ \xymatrix{
% (\R^{n})^{k+1} \x (\R^{n})^{k+1}\ar[d]_{z_{k+1}} \ar[rr]^-{D(U_k(f))} & & \R^{m}
% \\
% (\R^{n} \x \R^{n})^{k+1}\ar@/_2ex/[urr]_-{\;U_k(\seqD f)}
% } \]

%  $$U_k(\seqD f)\circ z_k = D(U_k(f))$$

% \noindent where $z_k: (\R^n)^k\x(\R^n)^k \to (\R^n\x\R^n)^k$ is the
% ``zipping map''
%   $$z_k(\v{\Delta x_1}, \v{\Delta x_2}, \ldots, \v{\Delta x_k}, \v{x_1}, \v{x_2}, \ldots, \v{x_k}) 
%   = (\v{\Delta x_1}, \v{x_1}, \v{\Delta x_2}, \v{x_2}, \ldots, \v{\Delta x_k}, \v{x_k})$$
%   for all $\v{\Delta x_i}, \v{x_i} \in \R^n$, and $D$ is the curried derivative
%   operation from classical analysis.

Note that this linear causal function can be evaluated at a sequence $\Delta
\sigma \in (\R^n)^\w$ by multiplying the infinite matrix by $\Delta \sigma$,
considered as an infinite column vector.

% We think of the causal function $f$ taking a sequence $\sigma$ to a
% sequence $\tau$ via $f(\sigma) = \tau$. This formulation of derivative
% allows us to think of $\seqD f$ as taking two sequences $\seqD
% f(\Delta \sigma, \sigma) = \Delta \tau$, where the $\Delta \sigma$
% argument is a sequence of small changes made to the ``base sequence''
% $\sigma$ and the output $\Delta \tau$ is the best linear approximation
% available to the the change in output of $f$.

\subsection{Examples}

Next, we use this definition of derivative to find the causal derivatives of
some basic functions from Rutten's stream calculus.

\begin{example}\label{ex:streamSumDiff}
We show the pointwise sum stream function $+: (\R^2)^\w \to \R^\w$ is its own
derivative at every point $(\sigma, \tau) \in (\R^2)^\w$. Note
${U_k(+)(\sigma_0, \tau_0, \ldots, \sigma_k, \tau_k) = \sigma_k + \tau_k}$, so
${J(U_k(+))(\sigma_0, \tau_0, \ldots, \sigma_k, \tau_k) =
\begin{bmatrix} 0 & \ldots & 0 & 1 & 1\end{bmatrix}}$. This is the matrix
representation of $U_k(+)$ itself, so $(\seqD +)(\sigma, \tau) = +$ or, in other
notation, $(\seqD +)(\sigma, \tau)(\Delta\sigma, \Delta\tau) = \Delta\sigma + \Delta\tau$
for any $\sigma, \tau, \Delta\sigma, \Delta\tau \in \R^\w$.

This argument can be repeated for all pointwise sum functions
$+:(\R^n\x\R^n)^\w \to (\R^n)^\w$, replacing the ``1'' blocks in the Jacobian
above with $I_n$.
\end{example}

Since the derivative of any constant $x: 1 \to \R^n$ is $0_{\R^n}: 1 \to
\R^n$, the derivative of any constant sequence must necessarily be the zero
sequence. In stream calculus, there are two important constant sequences
defined corecursively: $[r]$ defined by $\hd([r])(*) = r$ and $\partial_*([r])
= [0]$ for all $r \in \R$ and $X$ defined by $\hd(X)(*) = 0$ and
$\partial_*(X) = [1]$. Written out as sequences, $[r] = (r, 0, 0, 0,
\ldots)$ and $X = (0, 1, 0, 0, \ldots)$.

\begin{example}\label{ex:constantElementDiff}
  $\seqD [r] = \seqD X = [0]$.
\end{example}

Next, we consider the Cauchy sequence product. Under the correspondence
between sequences $\sigma \in \R^\w$ and formal power series $\sum \sigma_ix^i
\in \R[[x]]$, the Cauchy product is the sequence operation corresponding to
the (Cauchy) product of formal power series. This operation is coalgebraically
characterized in Rutten~\cite{Rutten_2005} as the unique function $\x:
(\R^2)^\w \to \R^\w$ satisfying $\hd(\x)(s_0, t_0) = s_0\cdot t_0$ and
$(\partial_{(s_0, t_0)}\x)(\sigma, \tau) =
\tl(\sigma)\x\tau + [s_0]\times\tl(\tau)$. For our purposes, the explicit
definition is more useful: $U_k(\x)(\sigma_{0:k}, \tau_{0:k}) =
\sum_{i=0}^k\sigma_i\cdot\tau_{k-i}$.

\begin{example}\label{ex:cauchyProductDiff}
  We compute the derivative of the Cauchy product.
  $$J(U_k(\x))(\sigma_0, \tau_0, \ldots, \sigma_k, \tau_k) =
  \begin{bmatrix}\tau_{k} & \sigma_k & \tau_{k-1} & \sigma_{k-1} & \ldots & \tau_0 & \sigma_0\end{bmatrix}$$
  Notice that multiplying this matrix by (an initial segment) of a small change sequence
  $(\Delta\sigma_0, \Delta\tau_0, \ldots, \Delta\sigma_k, \Delta\tau_k)$ yields
  $$J(U_k(\x))(\sigma_0, \tau_0, \ldots, \sigma_k, \tau_k)
  (\Delta\sigma_0, \Delta\tau_0, \ldots, \Delta\sigma_k, \Delta\tau_k) 
  = \sum_{i=0}^k \Delta\sigma_i\cdot\tau_{k-i} + \sum_{i=0}^k \sigma_i\cdot\Delta\tau_{k-i}$$

  Therefore,
  $({\seqD \x}(\sigma, \tau))(\Delta\sigma, \Delta\tau) = \Delta\sigma\x\tau + \sigma\x\Delta\tau$.
\end{example}

Another sequence product considered in the stream calculus is the Hadamard
product, also called the pointwise product. Defined coalgebraically, the
Hadamard product is the unique binary operation defined by $\hd(\odot)(s_0,
t_0) = s_0\cdot t_0$ and $(\partial_{(s_0, t_0)}\odot)(\sigma, \tau) =
\tl(\sigma)\odot\tl(\tau)$. This has a similar derivative to the Cauchy
product: ${\seqD \odot}(\sigma, \tau)(\Delta\sigma, \Delta\tau) =
\Delta\sigma\odot\tau + \sigma\odot\Delta\tau$.

Note that these derivatives make sense without any reference to properties of
the sequences used. We are not aware of a way to realize this derivative as an
instance of a notion of derivative known in analysis. The most obvious notion
to try is a Fréchet derivative induced by a norm on the space of sequences.
However, all norms we know on these spaces, including $\ell^p$-norms and
$\gamma$-geometric norms $\|\sigma\| = \sum
\sigma_i\cdot\gamma^i$ for $\gamma \in (0, 1]$, restrict the space of
sequences to various extents.

%%%%%%%%%%%%%%%%%%%%%%%%%%%%%%%%%%%%%%%%%%%%%%%%%%%%%%%%%%%%%%%%%%%%%%%
%% SECTION: Rules of causal differentiation                          %%
%%%%%%%%%%%%%%%%%%%%%%%%%%%%%%%%%%%%%%%%%%%%%%%%%%%%%%%%%%%%%%%%%%%%%%%
%%%%%%%%%%%%%%%%%%%%%%%%%%%%%%%%%%%%%%%%%%%%%%%%%%%%%%%%%%%%%%%%%%%%%%%
\section{Rules of causal differentiation}\label{sec:rules}

Just as it is impractical to compute all derivatives from the definition in
undergraduate calculus, it is also impractical to compute causal derivatives
directly from the definition. To ease this burden, one typically proves
various ``rules'' of differentiation which provide compositional recipes for
finding derivatives. That is our task in this section.

There are at least two good reasons to hope \emph{a priori} that the standard
rules of differentiation might hold for causal derivatives. First, causal
derivatives were defined to agree with standard derivatives in their finite
approximants. Since these approximant derivatives satisfy these rules, we
might hope that they hold over the limiting process. Second, \textbf{smooth}
causal functions form a Cartesian differential category, as was shown in
\cite{sprungerLICS2019}. The theory of Cartesian differential categories
includes as axioms or theorems abstract versions of the chain rule, sum rule,
etc. However, neither of these reasons are immediately sufficient, so we
must provide independent justification.

\subsection{Basic rules and their consequences}

We begin by stating some rules familiar from undergraduate calculus.

\begin{proposition}[causal chain rule]\label{prop:causalChainRule}
  Suppose $f: (\R^n)^\w \to (\R^m)^\w$ and $g: (\R^m)^\w \to (\R^\ell)^\w$ are
  causal functions. Suppose further $f$ is differentiable at $\sigma\in(\R^n)^\w$
  and $g$ is differentiable at $f(\sigma)$. Then $h = g \circ f$ is differentiable at
  $\sigma$ and its derivative is $\seqD g(f(\sigma))\circ \seqD f(\sigma)$.
\end{proposition}
\begin{proof}
  Let $f_k = T_k(f)$, $g_k = T_k(g)$, and $h_k = T_k(h)$. We know $h_k = g_k
  \circ f_k$. We show the stringwise approximants of $\seqD (g\circ
  f)(\sigma)$ and $\seqD g(f(\sigma))\circ \seqD f(\sigma)$ match.%
  \begin{align*}
    T_k(\seqD (g\circ f)(\sigma)) &= J(h_k)(\sigma_{0:k})
    = J(g_k \circ f_k)(\sigma_{0:k})\\
    &= J(g_k)(f_k(\sigma_{0:k}))\x J(f_k)(\sigma_{0:k}) &(*)\\
    &= J(g_k)(f(\sigma)_{0:k})\x J(f_k)(\sigma_{0:k}) \\
    &= T_k(\seqD g(f(\sigma))) \circ T_k(\seqD f(\sigma))
    = T_k(\seqD g(f(\sigma)) \circ \seqD f(\sigma))
  \end{align*}
  where the starred line is by the classical chain rule.
\end{proof}

Since we have already overloaded $\x$ for both Cauchy stream product and
matrix product, we use $\|$ for the parallel composition of functions, where
the parallel composition of $\phi: \R^n \to \R^m$ and $\psi: \R^p\to \R^q$ is
$\phi\| \psi: \R^{n+p}\to\R^{m+q}$ defined by $(\phi\| \psi)(x, y) = (\phi(x),
\psi(y))$ for $x \in \R^p$ and $y \in \R^p$. We do not know of a standard name
for this rule, but in multivariable calculus there is a rule $J(\phi\|
\psi)(x, y) = J\phi(x)\| J\psi(y)$, which we shall call the parallel rule.
There is a similar rule for causal derivatives we describe next.

\begin{proposition}[causal parallel rule]\label{prop:causalParallelRule}
  Suppose $f: (\R^n)^\w \to (\R^m)^\w$ and $h: (\R^p)^\w \to (\R^q)^\w$ are
  causal functions, and that they are differentiable at $\sigma \in (\R^n)^\w$ and $\tau
  \in (\R^p)^\w$, respectively. Then $f\| h: (\R^{n+p})^\w \to (\R^{m+q})^\w$
  is differentiable at $(\sigma, \tau) \in (\R^{n+p})^\w$ and its derivative
  is $\seqD f(\sigma) \| \seqD h(\tau)$.
\end{proposition}
\begin{proof}
  The stringwise approximants of $\seqD(f\| h)(\sigma,\tau)$ and $\seqD
  f(\sigma) \| \seqD h(\tau)$ match:
  \begin{align*}
    T_k(\seqD (f\| h)(\sigma, \tau)) &= J(T_k(f\| h))(\sigma_{0:k}, \tau_{0:k})
    = J(T_k(f)\| T_k(h))(\sigma_{0:k}, \tau_{0:k}) \\
    &= J(T_k(f))(\sigma_{0:k}) \| J(T_k(h))(\tau_{0:k}) &(*)\\
    &= T_k(\seqD f(\sigma)) \| T_k(\seqD h(\tau))
    = T_k(\seqD f(\sigma) \| \seqD h(\tau))
  \end{align*}
  where the starred line is by the classical parallel rule.
\end{proof}

\begin{proposition}[causal linearity]\label{prop:causalLinearRule}
  If $f: (\R^n)^\w \to (\R^m)^\w$ is a linear causal function, it is
  differentiable at every $\sigma \in (\R^n)^\w$ and its derivative is
  $\seqD f(\sigma) = f$.
\end{proposition}

These three results are the fundamental properties of causal differentiation
we will be using. Many other standard rules are consequences of these.
For example, we can derive a sum rule from these properties.

\begin{definition}\label{def:sumOfMaps}
  The \emph{sum} of two causal maps $f, g: (\R^n)^\w \to (\R^m)^\w$ is defined
  to be $f + g \teq {+ \circ (f\| g) \circ \Delta_{(\R^n)^\w}}$, where
  $\Delta_{(\R^n)^\w}$ is the sequence duplication map.
\end{definition}

\begin{proposition}[causal sum rule]\label{prop:causalSumRule}
  If $f$ and $g$ as in Definition~\ref{def:sumOfMaps} are both differentiable
  at $\sigma$, so is their sum and its derivative is $\seqD f(\sigma) + \seqD g(\sigma)$.
\end{proposition}
\begin{proof}
  Using the properties above, we find
  \begin{align*}
    \seqD(f + g)(\sigma) &= \seqD(+ \circ (f\| g) \circ \Delta_{(\R^n)^\w})(\sigma) & \text{(sum of maps def'n)} \\
    &= \seqD(+)((f\| g\circ\Delta_{(\R^n)^\w})(\sigma)) \circ \seqD(f\| g\circ \Delta_{(\R^n)^\w})(\sigma) & \text{(causal chain rule)} \\
    &= + \circ \seqD(f\| g\circ \Delta_{(\R^n)^\w})(\sigma) & \text{(linearity of +)} \\
    &= + \circ \seqD(f\| g)(\Delta_{(\R^n)^\w}(\sigma)) \circ \seqD(\Delta_{(\R^n)^\w})(\sigma) & \text{(causal chain rule)} \\
    &= + \circ \seqD(f\| g)(\sigma, \sigma) \circ \Delta_{(\R^n)^\w} & \text{(def'n \& linearity of $\Delta$)} \\
    &= + \circ (\seqD f(\sigma)\| \seqD g(\sigma)) \circ \Delta_{(\R^n)^\w} & \text{(causal parallel rule)} \\
    &= \seqD f(\sigma) + \seqD g(\sigma) & \text{(sum of maps def'n)}
  \end{align*}
  as desired.
\end{proof}

For functions $f, g: \R^\w \to \R^\w$, we can define their Cauchy and Hadamard
products $f\x g$ and $f\odot g$ with the pattern of
Definition~\ref{def:sumOfMaps} and prove two product rules using the
derivatives of the binary operations $\x$ and $\odot$ we computed earlier.

\begin{proposition}[causal product rules]\label{prop:causalProductRule}
  If $f, g: \R^\w \to \R^\w$ are causal functions differentiable at $\sigma$, so
  are their Cauchy and Hadamard products, and their derivatives are
  \begin{align*}
    \seqD(f\x g)(\sigma)(\Delta\sigma) &= \seqD f(\sigma)(\Delta \sigma)\x g(\sigma) + f(\sigma)\x\seqD g(\sigma)(\Delta \sigma) \\
    \seqD(f\odot g)(\sigma)(\Delta\sigma) &= \seqD f(\sigma)(\Delta \sigma)\odot g(\sigma) + f(\sigma)\odot\seqD g(\sigma)(\Delta \sigma)
  \end{align*}
\end{proposition}

A typical point of confusion in undergraduate calculus is the role of
constants: sometimes they are treated like elements of the underlying vector
space and sometimes like functions which always return that vector. In our
calculus, a constant can similarly sometimes mean a fixed sequence picked out
by $c: 1 \to (\R^n)^\w$ or the composition of this map after a discarding map
$\final{(\R^n)^\w}: (\R^n)^\w \to 1$. We have described the derivative of a
constant element in Example~\ref{ex:constantElementDiff}, now we treat
constant maps.

\begin{proposition}[causal constant rule]
  The derivative of $\final{(\R^n)^\w}: (\R^n)^\w \to 1$ is $\final{(\R^n)^\w}$.
  If $c: (\R^n)^\w \to (\R^m)^\w$ is a constant map, its derivative is the
  constant map $[0](\sigma) \equiv 0_{(\R^m)^\w}$.
\end{proposition}

\begin{proposition}[causal constant multiple rule]
  If $c: \R^\w \to \R^\w$ is a constant function and $f: \R^\w \to
  \R^\w$ is any other causal function differentiable at $\sigma$, so is
  $c\x f$ and its derivative is $c \x \seqD f(\sigma)$.
\end{proposition}
\begin{proof}
Combine the causal product rule and the causal constant rule.
\end{proof}

\subsection{Implicit causal differentiation}

We have seen the standard rules presented in the last section are useful as
computational shortcuts, just as they are in undergraduate calculus. In the
causal calculus they turn out to be perhaps even more crucial, since some
differentiable causal functions do not have simple closed forms, so trying
to find their derivative from the definition is extremely difficult.

The \emph{stream inverse}~\cite{Rutten_2005} is the first partial causal
function we will consider. This operation is defined on $\sigma\in \R^\w$ such
that $\sigma_0 \neq 0$ with the unbounded-order recurrence relation%
$$[\sigma\inv]_k = \begin{cases} \frac{1}{\sigma_0} & \text{ if } k = 0 \\
  -\frac{1}{\sigma_0}\cdot\displaystyle{\sum_{i = 0}^{k-1}}\left(\sigma_{n-i}\cdot 
  [\sigma\inv]_{i}\right) & \text{ if } k > 0
\end{cases}.$$

Reasoning about this function in terms of its components is extraordinarily
difficult since each component is defined in terms of all the preceding
components. However, there is a useful fact from Rutten~\cite{Rutten_2005}
which we can use to find the derivative of this operation at all $\sigma$
where it is defined: $\sigma \x \sigma\inv = [1]$.

\begin{proposition}[causal reciprocal rule]\label{prop:causalReciprocalRule}
  The partial function $(\cdot)\inv: \R^\w\to\R^\w$ is differentiable at
  all $\sigma \in \R^\w$ such that $\sigma_0 \neq 0$, and its derivative
  is $$(\seqD(\cdot)\inv)(\sigma)(\Delta\sigma) = [-1]\x\sigma\inv\x\sigma\inv\x\Delta\sigma$$
\end{proposition}
\begin{proof} Since $\sigma \x \sigma\inv = [1]$, their derivatives must also be equal. In particular:
  $$[0] = \seqD [1] = \seqD(\sigma\x\sigma\inv)(\Delta\sigma) 
    = \sigma\x(\seqD(\cdot)\inv)(\sigma)(\Delta\sigma) + \Delta\sigma\x(\sigma\inv)$$
  using the causal product rule. Solving this equation for 
  $(\seqD(\cdot)\inv)(\sigma)(\Delta\sigma)$ yields
  $$(\seqD(\cdot)\inv)(\sigma)(\Delta\sigma) =  [-1]\x\sigma\inv\x\sigma\inv\x\Delta\sigma$$
  where we are implicitly using many of the identities established in~\cite{Rutten_2005}.
\end{proof}

When adopting the conventions that $\sigma^{-n} \teq \sigma^{-(n-1)} \x
\sigma\inv$ and $\sigma\x\tau\inv \teq \frac{\sigma}{\tau}$, this rule
looks quite like the usual rule for the derivative of the reciprocal
function: $(J(\cdot)\inv)(x)(\Delta x) = -\frac{\Delta x}{x^2}$.

\begin{proposition}[causal quotient rule]\label{prop:causalQuotientRule}
  If $f, g: \R^\w \to \R^\w$ are causal functions differentiable at $\sigma$ and
  $g(\sigma)_0 \neq 0$, then $\frac{f}{g}$ is also differentiable at $\sigma$ and
  its derivative is
  $$\frac{\seqD f(\sigma)(\Delta\sigma)\x g(\sigma) + [-1]\x f(\sigma)\x\seqD g(\sigma)(\Delta\sigma)}{g(\sigma)^2}.$$
\end{proposition}

% A similar application of implicit differentiation allows us to establish
% the following.

% \begin{proposition}[causal inverse rule]
%   Suppose $f: (\R^n)^\w \to (\R^n)^\w$ is a causal function differentiable at $\sigma$,
%   such that both $f$ and $\seqD f(\sigma)$ have causal inverse functions. Say the
%   inverse of $f$ is $g$ and the inverse of $\seqD f(\sigma)$ is $h$. 
%   Then $g$ is differentiable at $f(\sigma)$ and has derivative $h$.
% \end{proposition}

\subsection{The recurrence rule}

So far, causal differential calculus is rather similar to traditional
differential calculus. There are two different product rules corresponding to
two different products. We were forced to use an implicit differentiation
trick to find the derivative of the reciprocal function, but in the end we
found a familiar result. However, next we state a rule with no traditional
analogue.

\begin{theorem}[causal recurrence rule]\label{thm:causalRecurrenceRule}
  Let $g: \R^n\x\R^m \to \R^m$ be differentiable (everywhere) and $i \in
  \R^m$. Then $\rec_i(g): (\R^n)^\w \to (\R^m)^\w$ is
  differentiable (everywhere) as a causal function and its derivative
  $\Delta\tau \teq [\seqD \rec_i(g)](\sigma)(\Delta\sigma)$ satisfies
  the following recurrence: $$\begin{cases}
    \tau_{k+1} = g(\sigma_{k+1}, \tau_k) &\text{ after } \tau_0 = g(\sigma_0, i) \\
    \Delta\tau_{k+1} = Jg(\sigma_{k+1}, \tau_k)(\Delta\sigma_{k+1}, \Delta\tau_k) &\text{ after }
    \Delta\tau_0 = Jg(\sigma_0, i)(\Delta\sigma_0, 0_{\R^m})
  \end{cases}$$
\end{theorem}
\begin{proof}
We check $U_k(\seqD \rec_i(g)(\sigma))(\Delta\sigma_{0:k}) = \Delta\tau_k$ by induction on $k$.
To simplify our notation, we write $u_k \teq U_k(\rec_i(g))$.
The base case is easy:
\begin{align*}
  U_0([\seqD \rec_i(g)](\sigma))(\Delta\sigma_0) &= J(U_0(\rec_i(g)))(\sigma_0)(\Delta\sigma_0) \\
  &= J(\lambda x.g(x, i))(\sigma_0)(\Delta\sigma_0) = Jg(\sigma_0, i)(\Delta\sigma_0, 0_{\R^m})
\end{align*}

The induction step uses the fact that
$u_k(\sigma_{0:k}) = g(\sigma_k, u_{k-1}(\sigma_{0:k-1}))$.
\begin{align*}
  U_k([\seqD \rec_i(g)](\sigma))(\Delta\sigma_{0:k}) &= Ju_k(\sigma_{0:k})(\Delta\sigma_{0:k}) \\
  &= [Jg(\sigma_k, \tau_{k-1})\circ 
  \tuple{J\pi_k(\sigma_{0:k}), J(u_{k-1}\circ \overline{\pi_k})(\sigma_{0:k})}](\Delta\sigma_{0:k}) \\
  &= [Jg(\sigma_k, \tau_{k-1})\circ\tuple{\pi_{k}, Ju_{k-1}(\sigma_{0:k-1})\circ\overline{\pi_k}}](\Delta\sigma_{0:k}) \\
  &= Jg(\sigma_k, \tau_{k-1})(\Delta\sigma_{k}, Ju_{k-1}(\sigma_{0:k-1})(\Delta\sigma_{0:k-1})) \\
  &= Jg(\sigma_k, \tau_{k-1})(\Delta\sigma_{k}, \Delta\tau_{k-1})
\end{align*}
where $\overline{\pi_k}$ is the map discarding the last element of a list.
\end{proof}

Degenerate recurrences, which do not refer to previous values generated by the
recurrence, are a special instance of this rule.

\begin{corollary}[causal map rule]
  Let $h: \R^n \to \R^m$ be a differentiable function. Then $\map(h)$ is 
  differentiable as a causal function, and its derivative is $\map(Jh)$.
\end{corollary}

To illustrate the recurrence rule, we revisit the running product function,
introduced in Example~\ref{ex:runningProduct}, and compute its derivative.

\begin{example}
  The unary running product function $\prod: \R^\w \to \R^\w$ was defined to
  be $\rec_1(g)$ where $g$ is binary multiplication of reals. In approximant
  form, $U_k(\rec_1(g))(\sigma_{0:k}) = \prod_{i=0}^k \sigma_i$. We compute a
  recurrence for the derivative of this function using the recurrence rule.

  Since $g$ is binary multiplication, $Jg(s, t)(\Delta s, \Delta t) = \Delta s
  \cdot t + s\cdot \Delta t$. By the recurrence rule, $[\seqD
  \rec_i(g)](\sigma)(\Delta\sigma)$ satisfies the recurrence
  $$\begin{cases}
    \tau_{k+1} = \sigma_{k+1}\cdot\tau_k 
    &\text{ after } \tau_0 %= \sigma_0\cdot 1  
    = \sigma_0\\
    \Delta\tau_{k+1} %= Jg(\sigma_{k+1}, \tau_k)(\Delta\sigma_{k+1}, \Delta\tau_k)
    = \Delta\sigma_{k+1}\cdot\tau_k + \sigma_{k+1}\cdot\Delta\tau_k 
    &\text{ after } \Delta\tau_0 %= Jg(\sigma_0, 1)(\Delta\sigma_0, 0) 
    = \Delta\sigma_0
  \end{cases}$$

  Note that a direct computation of the derivative of this function is
  available since we have a simple form for its pointwise approximants.
  Directly from the definition we would get
  $$\Delta\tau_k = U_k(\seqD \rec_1(g)(\sigma))(\Delta\sigma_{0:k}) 
  = \sum_{i = 0}^k \prod_{j = 0}^k \rho_{ij}$$
  where $\rho_{ij}$ is $\sigma_j$ if $i \neq j$ and $\Delta\sigma_j$ otherwise.

  Used naively, this formula results in $O(k^2)$ real number multiplications,
  and requires access to the entire initial segment of $\sigma$ at all times.
  In contrast, computing the same quantity using the recurrence obtained by
  the recurrence rule requires $O(k)$ multiplications and can be computed
  on-the-fly, requiring only the availability of the first elements of
  $\sigma$ and $\Delta\sigma$ to make initial progress and releasing their
  memory just after use.
\end{example}

\section{An extended example: Elman networks}\label{sec:elman}

We next turn toward a potential application domain of our causal differential
calculus: machine learning. In particular, we demonstrate that it is possible
to use this calculus in the training of recurrent neural networks (RNNs). RNNs
differ from the more common feedforward network in that they are designed to
process sequences of inputs rather than single inputs. This makes them
especially useful in analyzing long texts (sequences of words), spoken
language (sequences of sounds), and videos (sequences of images). In fact,
particular RNN architectures are the core underlying technologies of many
speech recognition products today, such as Alexa and Siri.

In this section, we will be using our causal differential calculus to find the
derivative of a simple kind of recurrent neural network, namely an Elman
network~\cite{Elman_1990}. This is an influential early example of a network
with feedback, though modern feedback networks typically have more structure.
Elman networks can operate on sequences of vectors from $\R^n$, but to keep
things slightly simpler we will consider Elman networks operating on sequences
of real numbers only.

Let $\alpha, \beta, \gamma, \delta, \epsilon \in \R$ be arbitrary parameters
and $\phi_1, \phi_2: \R\to\R$ be arbitrary differentiable ``activation''
functions.\footnote{``Activation'' here has no technical meaning, but carries
a connotation that the function is likely taken from a folklore set of
functions including the sigmoid function, hyperbolic tangent, softplus,
rectified linear unit, and logistic function. Usually these functions have
bounded range, often $[0,1]$.}  Given an input sequence $\sigma \in \R^\w$,
the Elman network defined by these parameters produces the sequence $E(\sigma)
= \tau \in \R^\w$ satisfying the following recurrence: $$\begin{cases}
  \rho_{k+1} = \phi_1(\alpha\sigma_{k+1} + \beta\rho_k + \gamma) 
  &\text{ after } \rho_0 = \phi_1(\alpha\sigma_0 + \gamma) \\
  \tau_{k+1} = \phi_2(\delta\rho_{k+1} + \epsilon) 
  &\text{ after } \tau_0 = \phi_2(\delta\rho_0 + \epsilon)
\end{cases}$$

In our notation, if we define $g_1(x, y) \teq \phi_1(\alpha x + \beta y +
\gamma)$ and $g_2(x) \teq \phi_2(\delta x + \epsilon)$, then $E \teq \map(g_2)
\circ \rec_{0}(g_1)$. We can therefore find the causal derivative of
this Elman network relatively easily using the causal chain rule and causal
recurrence rule. Indeed, letting $\seqD E(\sigma)(\Delta\sigma) = \Delta\tau$,
these rules tell us $\Delta\tau$ satisfies the recurrence:
$$\begin{cases}
  \rho_{k+1} = \phi_1(\alpha\sigma_{k+1} + \beta\rho_k + \gamma) 
  &\text{ after } \rho_0 = \phi_1(\alpha\sigma_0 + \gamma) \\
  \tau_{k+1} = \phi_2(\delta\rho_{k+1} + \epsilon) 
  &\text{ after } \tau_0 = \phi_2(\delta\rho_0 + \epsilon) \\
  \Delta\rho_{k+1} = \phi_1'(\alpha\sigma_{k+1} + \beta\rho_k + \gamma)\cdot
  (\alpha\Delta\sigma_{k+1} + \beta\Delta\rho_k)
  &\text{ after } \Delta\rho_0 = \phi_1'(\alpha\sigma_0 + \gamma)\cdot(\alpha\Delta\sigma_0)\\
  \Delta\tau_{k+1} = \phi_2'(\delta\rho_{k+1} + \epsilon)\cdot(\delta\Delta\rho_{k+1})
  &\text{ after } \Delta\tau_0 = \phi_2'(\delta\rho_0 + \epsilon)\cdot(\delta\Delta\rho_0)
\end{cases}$$

This derivative tells us how we would expect the output of the Elman network
to change in response to a small change $\Delta\sigma$ to its input sequence
$\sigma$. This can be useful information in analyzing the behavior of the
network. However, we can also use causal differentiation to predict how the
network's output would change in response to a small change in one of the
\emph{parameters}, which is a crucial piece of information used when training
the network.

Let us now imagine that we have some data on how this Elman network
\emph{should} behave, in the form of an input/output pair $(\hat{\sigma},
\hat{\tau}) \in \R^\w\x\R^\w$ representing ground truth, and we want
to figure out how to adjust one of the parameters, say $\alpha$, so that our
Elman network better reflects this ground truth.

We can define a causal function related to the Elman network $E$, but where we
now consider $\alpha$ to be a variable and fix $\sigma$ to be $\hat\sigma$.
Denote this function $E_{\hat\sigma}: \R^\w\to\R^\w$ and note that if $\tau =
E_{\hat\sigma}(\hat\alpha)$ for $\hat\alpha \in \R^\w$, then $\tau$ satisfies the
recurrence relation $$\begin{cases}
  \rho_{k+1} = \phi_1(\alpha\hat\sigma_{k+1} + \beta\rho_k + \gamma)
  &\text{ after } \rho_0 = \phi_1(\alpha\hat\sigma_0 + \gamma) \\
  \tau_{k+1} = \phi_2(\delta\rho_{k+1} + \epsilon)
  &\text{ after } \tau_0 = \phi_2(\delta\rho_0 + \epsilon)
\end{cases}$$

We have simplified our expression using the fact that parameters are fixed
values that do not change in the course of the computation of the output
sequence, so $\hat\alpha_k = \alpha$ for all $k \in \w$. Similarly, when we make
small change to this parameter, that small change will remain independent of
the entry in the sequence, so $\widehat{\Delta\alpha}_k = \Delta\alpha$ for all $k$.

We can compute the derivative of this recurrence relation similarly to above,
and find it will satisfy the following recurrence relation:
$$\begin{cases}
  \rho_{k+1} = \phi_1(\alpha\hat\sigma_{k+1} + \beta\rho_k + \gamma)
  &\text{ after } \rho_0 = \phi_1(\alpha\hat\sigma_0 + \gamma) \\
  \tau_{k+1} = \phi_2(\delta\rho_{k+1} + \epsilon)
  &\text{ after } \tau_0 = \phi_2(\delta\rho_0 + \epsilon)\\
  \Delta\rho_{k+1} = \phi_1'(\alpha\hat\sigma_{k+1} + \beta\rho_k + \gamma)\cdot
  (\Delta\alpha\hat\sigma_{k+1} + \beta\Delta\rho_k)
  &\text{ after } \Delta\rho_0 = \phi_1'(\alpha\hat\sigma_0 + \gamma)\cdot(\Delta\alpha\hat\sigma_0) \\
  \Delta\tau_{k+1} = \phi_2'(\delta\rho_{k+1} + \epsilon)\cdot(\delta\Delta\rho_{k+1})
  &\text{ after } \Delta\tau_0 = \phi_2'(\delta\rho_0 + \epsilon)\cdot(\delta\Delta\rho_0)
\end{cases}$$

\begin{example}\label{ex:elmanSpecific}
  Let us take a very specific example to illustrate this process. We
  instantiate the above Elman network with $\alpha = \beta = \delta = 1$,
  $\gamma = 0.1$, $\epsilon = -0.1$ and $\phi_1 = \phi_2$ are both
  the sigmoid function.\footnote{The sigmoid function $\phi: \R\to\R$ is
  defined by $\phi(x) = \frac{1}{1 + e^{-x}}$. The sigmoid function is
  traditionally denoted by $\sigma$, but since we have been using $\sigma$ as
  a sequence variable we use $\phi$.}

  We suppose our ground truth data tells us a sequence starting $\hat\sigma =
  (1, 1, 1, 1, \ldots)$ should be sent to a sequence starting $\hat\tau =
  (0.60, 0.63, 0.63, 0.64, \ldots)$. In reality, our Elman network as
  currently parametrized sends $\hat\sigma$ to $(0.65707, 0.68226, 0.68503, 0.68533,
  \ldots)$, when rounded to 5 decimal places. Our task is to decide how to
  adjust $\alpha$ so that the new network will better match our data, in
  particular reducing every entry by about $0.05$.

  We begin by first writing out the recurrence relation for the derivative of
  $E_{\hat\sigma}$ from above with our particular choice of parameters. Since
  we have chosen many coefficients and all the entries of $\hat\sigma$ to be
  1, there is significant simplification: $$\begin{cases}
    \rho_{k+1} = \phi(\rho_k + 1.1) 
    &\text{ after } \rho_0 = \phi(1.1)\\
    \tau_{k+1} = \phi(\rho_{k+1} - 0.1)
    &\text{ after } \tau_0 = \phi(\rho_0 - 0.1) \\
    \Delta\rho_{k+1} = \phi'(\rho_k + 1.1)\cdot(\Delta\alpha + \Delta\rho_k)
    &\text{ after } \Delta\rho_0 = \phi'(1.1)\cdot\Delta\alpha \\
    \Delta\tau_{k+1} = \phi'(\rho_{k+1} - 0.1)\cdot\Delta\rho_{k+1}
    &\text{ after } \Delta\tau_0 = \phi'(\rho_0 - 0.1)\cdot\Delta\rho_0
  \end{cases}$$

  The only free variable in this recurrence is $\Delta\alpha$. We choose
  $\Delta\alpha = 0.1$, for reasons to be explained later. Then we can compute
  $\Delta\tau = (0.00422, 0.00302, 0.00265, 0.00259, \ldots)$.

  What does this tell us? The recurrence is supposed to compute the derivative
  of $E_{\hat\sigma}$ at 1 and apply the resulting linear map to 0.1. Using
  the interpretation of derivative as approximate change, this suggests that
  if we increase our parameter $\alpha$ from its current value of $1$ by
  $\Delta\alpha = 0.1$, we should expect $E_{\hat\sigma}(1.1)$ to be about
  $E_{\hat\sigma}(1) + (0.00422, 0.00302, 0.00265, 0.00259, \ldots)$. Since
  our goal is to \textbf{reduce} the output of the network, this adjustment is
  not a great idea.

  What are we to do? One option is to pick a new value for $\Delta\alpha$ and
  recompute the approximate change, but there is a smarter way. We know that
  the derivative of $E_{\hat\sigma}$ at 1 is linear, so if we instead
  \emph{decrease} $\alpha$ by 0.1, we would expect $E_{\hat\sigma}(0.9)$ to be
  about $E_{\hat\sigma}(1) - (0.00422, 0.00302, 0.00265, 0.00259, \ldots) =
  (0.65285, 0.67923, 0.68238, 0.68274, \ldots)$. Indeed, after making this
  adjustment, we find $E_{\hat\sigma}(0.9) = (0.65273, 0.67908, 0.68224,
  0.68261, \ldots)$. This adjustment ended up decreasing the result by about
  0.00015 more than we predicted, which amounts to approximately a 5\%
  overshot of the original prediction.

  While it is nice to know our prediction about the change was fairly
  accurate, subtracting 0.1 from $\alpha$ has not achieved our goal: in each
  component, our Elman network's output decreased by at most 0.005 while we
  were trying to create a reduction of 0.05. A natural idea here would be to
  \emph{really} exploit the linearity of the derivative and make a bigger
  adjustment to $\alpha$, namely subtracting
  $\frac{0.05}{0.005}\cdot\Delta\alpha = 10\cdot\Delta\alpha = 1$. Computing
  $E_{\hat\sigma}(0)$, we find it is actually $(0.60467, 0.63445, 0.64095,
  0.64235, \ldots)$, which is much closer to our goal than $E_{\hat\sigma}(0.9)$
  turned out to be.

  This seems like good news, but if we check the accuracy of the prediction
  our derivative makes, we would find that the actual reduction from
  $E_{\hat\sigma}(1)$ to $E_{\hat\sigma}(0)$ is between 25\% and 65\% greater
  than the derivative predicted. Thus, though we were able to make greater
  progress aligning our network with ground truth, the bigger adjustment came
  with much greater error. This is a classic tradeoff in neural network
  training: the linear approximation provided by the derivative is only valid
  locally, so taking bigger steps along the gradient comes with potentially
  greater rewards in terms of improvements in network performance but also
  carries extra risk that greater error could lead the training astray.
\end{example} 

\section{Conclusion, related work, and future directions}\label{sec:conclusion}

In this paper, we presented a basic differential calculus for causal functions
between sequences of real-valued vectors. We gave a definition of derivative
for causal functions, showed how to compute derivatives from this definition,
established many classical rules from multivariable calculus including the
chain, parallel, sum, product, reciprocal, and quotient rules. We additionally
showed a rule unique to the causal calculus: the recurrence rule. We then
showed how to use these rules in a practical example, namely the training of
an Elman network.

\emph{Related work.} We are not aware of other works directly treating
differentiation of causal functions, though we suspect there may be
connections to hard-core analysis literature. This work is obviously inspired
in results and structure by standard undergraduate multivariable calculus,
e.g.~\cite{spivak1965calculus}. We also have a related categorical treatment
of differentiation of causal functions~\cite{sprungerLICS2019} using the
framework of Cartesian differential categories~\cite{Blute09}. That is much
more abstract than the present work, but when concretized to the current
scenario would only apply to smooth causal functions.

Though we drew our example differentiable functions almost exclusively from
Rutten's stream calculus~\cite{Rutten_2005}, we would also like to point out
signal flow graphs as another interesting treatment of causal functions. an
interesting graphical representation of causal functions, investigated in
e.g.~\cite{Basold2014,DBLP:conf/concur/BonchiSZ14,DBLP:conf/popl/BonchiSZ15,
DBLP:conf/lics/Milius10}. We expect that interpreting our differential
calculus in this setting could yield a treatment of differentiation in string
diagrams.

We suspect recurrence rule we obtained, particularly when differentiating
Elman networks, may also have connections to the automatic differentiation
literature we are not aware of at this time. In particular, it does rather
seem like the recurrence rule augments a recurrence with dual numbers.

\emph{Future directions.} As neural networks become more advanced and
practitioners find new and interesting ways of using gradients of these
networks, we believe theoreticians have a role to play in systematizing the
theory of these new applications of derivatives. We believe that the coalgebra
community, as experts with many tools for understanding programs operating on,
infinite data structures, are particularly well-positioned to help develop
these theories. For example, nearly every rule of causal differentiation we
established here relies on a coalgebraically-derived property from Rutten's
stream calculus~\cite{Rutten_2005}. We looked at functions on sequences in
particular, but we have every reason to believe further results are possible
for more advanced neural network architectures on more exotic infinite data
structures.

We are particularly interested in merging our results here with a line of
research initiated in~\cite{sprungerLICS2019} using Cartesian differential
categories. We believe this causal calculus could be an instance of a
Cartesian differential \emph{restriction}
category~\cite{cockett2011differential}, which would drastically improve the
scope of our previous results to cover partial and non-smooth causal
functions.

\bibstyle{plain}
\bibliography{bib}

\begin{thebibliography}{10}

\bibitem{Basold2014}
Henning Basold, Marcello Bonsangue, Helle~Hvid Hansen, and Jan Rutten.
\newblock {\em (Co)Algebraic Characterizations of Signal Flow Graphs}, pages
  124--145.
\newblock Springer International Publishing, Cham, 2014.
\newblock URL: \url{https://doi.org/10.1007/978-3-319-06880-0_6}, \href
  {http://dx.doi.org/10.1007/978-3-319-06880-0_6}
  {\path{doi:10.1007/978-3-319-06880-0_6}}.

\bibitem{Blute09}
R~F Blute, J~R~B Cockett, and R~A~G Seely.
\newblock Cartesian differential categories.
\newblock {\em Theory and Applications of Categories}, 22:622--672, 2009.

\bibitem{DBLP:conf/concur/BonchiSZ14}
Filippo Bonchi, Pawe{\l} Soboci{\'{n}}ski, and Fabio Zanasi.
\newblock A categorical semantics of signal flow graphs.
\newblock In {\em {CONCUR} 2014 - Concurrency Theory - 25th International
  Conference, {CONCUR} 2014, Rome, Italy, September 2-5, 2014. Proceedings},
  pages 435--450, 2014.
\newblock URL: \url{https://doi.org/10.1007/978-3-662-44584-6\_30}, \href
  {http://dx.doi.org/10.1007/978-3-662-44584-6\_30}
  {\path{doi:10.1007/978-3-662-44584-6\_30}}.

\bibitem{DBLP:conf/popl/BonchiSZ15}
Filippo Bonchi, Pawe{l} Soboci{\'{n}}ski, and Fabio Zanasi.
\newblock Full abstraction for signal flow graphs.
\newblock In {\em Proceedings of the 42nd Annual {ACM} {SIGPLAN-SIGACT}
  Symposium on Principles of Programming Languages, {POPL} 2015, Mumbai, India,
  January 15-17, 2015}, pages 515--526, 2015.
\newblock URL: \url{https://doi.org/10.1145/2676726.2676993}, \href
  {http://dx.doi.org/10.1145/2676726.2676993}
  {\path{doi:10.1145/2676726.2676993}}.

\bibitem{cockett2011differential}
JRB Cockett, GSH Cruttwell, and JD~Gallagher.
\newblock Differential restriction categories.
\newblock {\em Theory and Applications of Categories}, 25(21):537--613, 2011.

\bibitem{Elman_1990}
Jeffrey~L. Elman.
\newblock Finding structure in time.
\newblock {\em Cognitive Science}, 14(2):179–211, Mar 1990.
\newblock \href {http://dx.doi.org/10.1207/s15516709cog1402_1}
  {\path{doi:10.1207/s15516709cog1402_1}}.

\bibitem{DBLP:conf/lics/Milius10}
Stefan Milius.
\newblock A sound and complete calculus for finite stream circuits.
\newblock In {\em Proceedings of the 25th Annual {IEEE} Symposium on Logic in
  Computer Science, {LICS} 2010, 11-14 July 2010, Edinburgh, United Kingdom},
  pages 421--430, 2010.
\newblock URL: \url{https://doi.org/10.1109/LICS.2010.11}, \href
  {http://dx.doi.org/10.1109/LICS.2010.11} {\path{doi:10.1109/LICS.2010.11}}.

\bibitem{sdes}
Jan Rutten, Clemens Kupke, and Helle~Hvid Hansen.
\newblock Stream differential equations: Specification formats and solution
  methods.
\newblock {\em Logical Methods in Computer Science}, 13, 2017.

\bibitem{Rutten_2005}
J.J.M.M. Rutten.
\newblock A coinductive calculus of streams.
\newblock {\em Mathematical Structures in Computer Science}, 15(1):93–147,
  Feb 2005.
\newblock \href {http://dx.doi.org/10.1017/S0960129504004517}
  {\path{doi:10.1017/S0960129504004517}}.

\bibitem{ruttenMealy}
J.J.M.M. Rutten.
\newblock Algebraic specification and coalgebraic synthesis of mealy automata.
\newblock {\em Electronic Notes in Theoretical Computer Science},
  160:305–319, Aug 2006.
\newblock \href {http://dx.doi.org/10.1016/j.entcs.2006.05.030}
  {\path{doi:10.1016/j.entcs.2006.05.030}}.

\bibitem{spivak1965calculus}
Michael Spivak.
\newblock Calculus on manifolds.
\newblock 1965.

\bibitem{sprungerLICS2019}
David Sprunger and Shin{-}ya Katsumata.
\newblock Differentiable causal computations via delayed trace.
\newblock {\em CoRR}, abs/1903.01093, 2019.
\newblock URL: \url{http://arxiv.org/abs/1903.01093}, \href
  {http://arxiv.org/abs/1903.01093} {\path{arXiv:1903.01093}}.

\bibitem{bptt}
P.~J. Werbos.
\newblock Backpropagation through time: what it does and how to do it.
\newblock {\em Proceedings of the IEEE}, 78(10):1550--1560, Oct 1990.
\newblock \href {http://dx.doi.org/10.1109/5.58337}
  {\path{doi:10.1109/5.58337}}.

\end{thebibliography}

\end{document}